\documentclass[conference,10pt,twocolumn,letter]{IEEEtran}
\IEEEoverridecommandlockouts
\ifCLASSINFOpdf
\else
\fi
\usepackage[cmex10]{amsmath}
\usepackage{amsthm}
\usepackage{amssymb}
\usepackage{mathrsfs}
\usepackage{amsbsy}
\usepackage{mathrsfs}
\usepackage{setspace}
\usepackage{graphicx}
\usepackage{caption}
\usepackage{subcaption}
\usepackage{psfrag}
\usepackage{newclude}
\usepackage[normalem]{ulem}
\usepackage{latexsym}
\usepackage[table]{xcolor}
\usepackage{multirow}
\usepackage{rotating}
\usepackage{booktabs}
\usepackage{amsmath}
\usepackage{mathrsfs}
\usepackage{bbm} % indicator function
\usepackage{filecontents}
\usepackage{amsmath,epsfig,cite,amsfonts,amssymb,psfrag,subfig}

\usepackage{accents}
\newlength{\dhatheight}

\allowdisplaybreaks

\newtheorem{theorem}{Theorem}
\newtheorem{lemma}{Lemma}

\graphicspath{{figures/}}
\allowdisplaybreaks

\begin{document}
\title{Optimal Energy Management for Energy Harvesting Transmitter and Receiver with Helper} \author{\IEEEauthorblockN{\normalsize
    Mohsen Abedi$^*$, Mohammad Javad Emadi$^*$, Behzad Shahrasbi$^\dagger$}
  \IEEEauthorblockA{\small $^*$Electrical Engineering Department, Amirkabir University of Technology, Tehran, Iran\\
  $^\dagger$Electrical Engineering and Computer Science Department, University of Central Florida\\
    E-mails: \{mohsenabedi, mj.emadi\}@aut.ac.ir, behzad@eecs.ucf.edu}
    \vspace{-.75cm}
    }
\maketitle
%%%%%%%%%%%%%%%%%%%%%%%%%%%%%%%%%%%%%%%%%%%%%%%%%%%%%%%%%%%%%%%%%%%%%%%%%%%%%%%%%%%%%%%%%%%%%%%%%%%%%%%%%%%%%%%%%%%%%%%%%%%%%%%%%%%%
\begin{abstract}
We study energy harvesting (EH) transmitter and receiver, where the receiver decodes data using the harvested energy from the nature and from an independent EH node, named helper. Helper cooperates with the receiver by transferring its harvested energy to the receiver over an orthogonal fading channel. We study an offline optimal power management policy to maximize the reliable information rate. The harvested energy in all three nodes are assumed to be known. We consider four different scenarios; First, for the case that both transmitter and the receiver have batteries, we show that the optimal policy is transferring the helper's harvested energy to the receiver, immediately. Next, for the case of non-battery receiver and full power transmitter, we model a virtual EH receiver with minimum energy constraint to achieve an optimal policy. Then, we consider a non-battery EH receiver and EH transmitter with battery. Finally, we derive optimal power management wherein neither the transmitter nor the receiver have batteries. We propose three iterative algorithms to compute optimal energy management policies. Numerical results are presented to corroborate the advantage of employing the helper.
\end{abstract}
\IEEEpeerreviewmaketitle
%%%%%%%%%%%%%%%%%%%%%%%%%%%%%%%%%%%%%%%%%%%%%%%%%%%%%%%%%%%%%%%%%%%%%%%%%%%%%%%%%%%%%%%%%%%%%%%%%%%%%%%%%%%%%%%%%%%%%%%%%%%%%%%%%%%%
\section{Introduction}
Green communications is a new concept which deals with using the harvested energy from the nature and efficiently transmit data over the communications networks. Thus, optimal harvested energy management policies have gained lots of interest in both theoretical and practical perspectives.

The optimal policy for the case where both source and destination harvest energy in a point-to-point data link is considered in \cite{31}. In \cite{32} and \cite{33}, authors find the optimal sampling rate to make tradeoff between sampling and decoding energy cost at EH receiver. In \cite{34}, authors study data link optimization for the case of channel state information at the receiver. In \cite{17}, energy cooperation between transmitter and two relays is considered. In \cite{30}, fading multiple access channel optimization with battery capacity and energy consumption constraints is studied. EH transmitter and relay with an energy arrival constraints is studied in \cite{13,14,15}. In\cite{ulu}, it is assumed that transmitter and the receiver rely exclusively on the energy harvested from the nature. The receiver uses the harvested energy for the \emph{decoding process}. The idea of energy cost of the processing for EH transmitter is also studied in \cite{gunduz}.
%An iterative algorithm for driving the Lagrangian multipliers in the problem with the set of linear constrains is presented in \cite{44}.

In this paper, we present EH transmitter and receiver with \emph{decoding cost} wherein there is an \emph{energy cooperation} link from EH \emph{helper} to the receiver. Utilizing energy link from the helper to the receiver, obviously increases the reliable data rate, especially when the harvested energy at the receiver is less than the required energy for the decoding process. We also assume that the \emph{helper has battery} to save the harvested energy. For each proposed EH scenario at transmitter and receiver, we propose an optimal energy management policy at the three nodes to maximize the reliable data rate. We consider the following four scenarios;

\begin{itemize}
\item \emph{EH transmitter and receiver with batteries}\\
We study the case where both transmitter and receiver are equipped with batteries. Thus they can save the harvested energy for the upcoming time slots. For this scenario, we analytically prove that the energy from the helper must be transferred to the receiver as soon as it is harvested from the nature. So, we present a closed form solution using the scheme given in \cite{ulu}.

\item \emph{Full power transmitter and non-battery receiver}\\
Assume that the transmitter has unlimited energy at all the time slots (full power) and the receiver is not equipped with a battery and has to consume all or part of the harvested energy at the moment. For this case, we propose an iterative algorithm based on a \emph{virtual} EH receiver with some minimum constraints on powers.

\item \emph{Transmitter with battery and non-battery EH receiver}\\
EH transmitter has a battery, so it then can store the harvested energy but the receiver is not equipped with a battery. So the receiver can only use its harvested energy at the current time slot. In this case, we decompose the solution to inner and outer problems and we prove that the optimal outer solution is derived iteratively by assuming that EH transmitter is full power.

\item \emph{EH transmitter and receiver with no batteries}\\
In this case, neither the transmitter nor the receiver has battery. So both of them can only use their harvested energy at the harvested time slot. For this case, we propose an iterative algorithm for optimal energy management.
\end{itemize}
For all scenarios, we present concrete proofs for optimality of all related iterative algorithms.

The rest of this paper is organized as follows. Section II describes the system model and problem statement. Section III presents problem statements and solutions. Section IV provides a numerical result and conclusions.

\begin{figure}[!t]
\centering
\psfrag{tx}[][][.8]{Transmitter}
\psfrag{m}[][][.8]{message}
\psfrag{H}[][][.8]{Helper}
\psfrag{D}[][][.8]{Receiver}
\psfrag{Z}[][][.8]{Noise}

\psfrag{a}[][][0.8]{$\alpha$}
\psfrag{dot}[][][.8]{$\cdots$}

\psfrag{e1}[][][.7]{$E_1$}
\psfrag{e2}[][][.7]{$E_2$}
\psfrag{e3}[][][.7]{$E_3$}
\psfrag{en}[][][.7]{$E_N$}

\psfrag{d1}[][][.7]{$\bar{E}_1$}
\psfrag{d2}[][][.7]{$\bar{E}_2$}
\psfrag{d3}[][][.7]{$\bar{E}_3$}
\psfrag{dn}[][][.7]{$\bar{E}_N$}

\psfrag{r1}[][][.7]{$H_1$}
\psfrag{r2}[][][.7]{$H_2$}
\psfrag{r3}[][][.7]{$H_3$}
\psfrag{rn}[][][.7]{$H_N$}

\includegraphics[width=8cm]{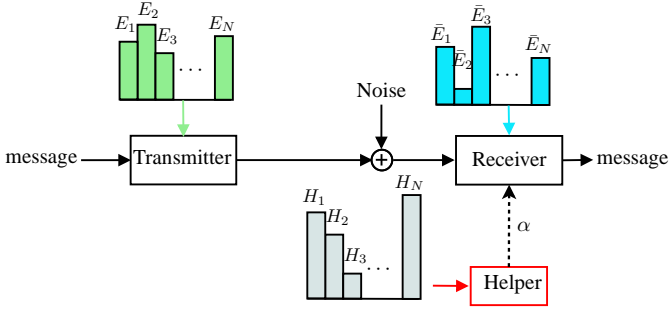}
%\centering
%\includegraphics[width=8.5cm]{fig1.png}
\caption{\small Cooperative transmission with energy harvesting transmitter, receiver and helper.}
\vspace{-0.6 cm}
\end{figure}

\section{System Model}
In this paper, we consider point-to-point communications where the transmitter sends and the receiver receives information by utilizing a helper, see Fig. 1. It is assumed that all the three nodes are energy harvesters and the receiver needs sufficient amount of power, known as \emph{decoding cost}, to decode data transmitted by the source. Although the receiver harvests energy from the nature, it can also receive energy from the helper at given time slots. However, the transmitter only harvests energy from the nature. Without loss
of generality, the presented method is applicable to the cases
where the helper transfers energy to the transmitter as well. The communications time is divided into $N$ equal time slots and the energy is harvested at the beginning of each time slot. At the $i^{th}$ slot, the harvested energies at the transmitter, receiver and the helper are denoted by ${\{E_i\}}_{i=1}^N$, ${\{\bar E_i\}}_{i=1}^N$, and ${\{H_i\}}_{i=1}^N$, respectively. We have assumed that at each time slot, the helper transfers some energy units to the receiver over a fading channel with energy efficiency $\alpha \in [0,1]$. In other words, when the helper sends $\delta_i$ units of its energy, the destination receives only $\alpha \times \delta_i$ amount of energy.

The channel between the transmitter and the receiver is AWGN one with zero-mean and unit variance noise, that is
 \begin{equation}
Y_i =X_i+Z_i,  \text{~for~}  i=1,...,N
\end{equation}
 where $i$ indicates the time slot, $X_i$ is the transmitted symbol with $\mathbb{E}(X_{i}^2)=p_i$, $Y_i$ is the received symbol at the destination and  $Z_i\thicksim\mathcal N(0,1)$. Since the receiver has no data buffer, it must decode the message at the end of the each time slot. Therefore, according to the normalized AWGN channel, the achievable information rate is $r_i=g(p_i)=\dfrac{1}{2} \log(1+p_i)$ \cite{cover}. Moreover, in order to decode the message at time slot $i$ in the receiver, $q_i=\varphi (r_i)$ units of energy in receiver's battery is used for processing. $\varphi(r_i)$ is the decoding cost function which is ``convex, monotone increasing in the incoming rate'' \cite{31}.
%\vspace{-0.1cm}

Throughout the paper we assume that there is offline information about the harvested energy from the nature by the transmitter, receiver and helper at any time slot, i.e., the harvested energy values ${\{E_i\}}_{i=1}^N$, ${\{\bar E_i\}}_{i=1}^N$, and ${\{H_i\}}_{i=1}^N$ are known beforehand. In the paper, we study optimal power management over the three nodes to increase the reliable information rate. In the rest of paper, we consider the four scenarios. For all cases, it is assumed that the helper can partially store its harvested energy for the upcoming time slots.

\section{Problem Statements}
\subsection{EH transmitter and receiver with batteries}
The problem for this case is formulated as
\vspace{-0.2cm}
 \begin{subequations}
    \small{\begin{align}
      & \  \underset{\boldsymbol{p},\boldsymbol{\delta}}{\max} \hspace{0.3cm} \sum_{i=1}^N g(p_i) \\
      & \  \text{s.t.} \hspace{0.3cm}\sum_{i=1}^j p_i\leq \sum_{i=1}^j E_i,\hspace{0.5cm} \forall j \\
      & \  \sum_{i=1}^j \varphi(g(p_i))\leq \sum_{i=1}^j \bar E_i + \alpha \delta_i,\hspace{0.5cm} \forall j \\
       & \  \sum_{i=1}^j \delta_i\leq \sum_{i=1}^j H_i \hspace{0.01cm},  \hspace{0.5cm} \forall j
    \end{align}}
  \end{subequations}
 In this problem (2b), (2c), and (2d) denote the transmitter and receiver and the helper energy causality constraints, respectively. It can be intuitively understood that transferring helper's energy instantaneously to the receiver is optimal, since the receiver has a battery to store the energy. We prove the claim in the following. Let define a new constraint $\sum_{i=1}^j \varphi(g(p_i))\leq \sum_{i=1}^j \bar E_i + \alpha H_i,\forall j$ as (2c$^\prime$).
\begin{lemma}
If $p_i^*$ and $\tilde{p}_i^*$ denote the optimum powers for (2,a-d) and for the new problem (2,a,b,c$^\prime$) respectively, then $p_i^*=\tilde{p}_i^*, \forall i$.
\end{lemma}
\begin{proof} Considering (2c) and (2d), we have $\sum_{i=1}^j \varphi(g(p_i))\leq \sum_{i=1}^j \bar E_i + \alpha \delta_i \leq \sum_{i=1}^j \bar E_i + \alpha H_i,\forall j$. Now , it can be argued that (2a,b,c$^\prime$,d) is the same as (2) but it is extended in receiver energy causality constraint. Moreover, we reduce (2a,b,c$^\prime$,d) to (2a,b,c$^\prime$) which is optimization problem independent from $\delta_i$. These two actions extend answer set and as a result, $\sum_{i=1}^N g(p_i^*) \leq \sum_{i=1}^N g(\tilde{p}_i^*)$. On the other hand, considering $\delta_j=H_j,\forall j$ as the special case, is feasible for (2) and it changes the problem to (2a,b,c$^\prime$). Therefore, $\sum_{i=1}^N g(\tilde{p}_i^*) \leq \sum_{i=1}^N g({p}_i^*)$. Now, we conclude that $\sum_{i=1}^N g(\tilde{p}_i^*) = \sum_{i=1}^N g({p}_i^*)$. Since the power policy $\{\tilde p_i^*\}_{i=1}^N$ is feasible for (2) and holds the objective function equal to that for $\{ p_i^*\}_{i=1}^N$, it can be concluded that $p_i^*=\tilde{p}_i^*, \forall i$ because of (2) convexity and the uniqueness of the optimal policy.
\end{proof}
  From lemma 1, we can substitute constraints in formats (2c) and (2d) with (2c$^\prime$) in case of problem convexity and other constraints independency from $\delta_i$. By replacing $r_i=g(p_i)$ in (2), the problem can be rewritten as
  \vspace{-0.2cm}
\begin{subequations}
\small{\begin{align}
      & \  \underset{\boldsymbol{r}}{\text{max}} \hspace{0.3cm} \sum_{i=1}^N r_i \\
      & \  \text{s.t.} \hspace{0.3cm}\sum_{i=1}^j g^{-1}(r_i)\leq \sum_{i=1}^j E_i,\hspace{0.5cm} \forall j \\
      & \  \sum_{i=1}^j \varphi(r_i)\leq \sum_{i=1}^j \bar E_i + \alpha H_i,\hspace{0.5cm} \forall j
\end{align}}
\end{subequations}
where the objective function is linear function and $g^{-1}(.)$ and $\varphi(.)$ are both convex functions. Thus, (3) is a convex optimization problem and the optimal policy is derived by use of the conventional methods \cite{boyd}. Now it is clear that the optimization problem (3) is similar to the problem considered in \cite[eq (3)]{ulu}. If $\{r_i^*\}_{i=1}^N$ are the optimal solutions of (3), using the same approach as \cite[lemma 1-3, and Theorem 1]{ulu}, we have the following two lemmas and Theorem 1.
\begin{lemma} $r_i^*\leq r_{i+1}^*$ for $i=1,...,N-1$.
\end{lemma}
\begin{lemma} if $r_{k}^*<r_{k+1}^*$ for some $k$, at least one of the constraints, (3b) or (3c), is satisfied with equality at $j=k$.
\end{lemma}
\begin{theorem} The optimal rates are
\vspace{-0.2cm}
\begin{equation}
\begin{aligned}
      & \  r_n^*= \min\{ g(\dfrac {\sum_{j=1}^{i_n}E_j-\sum_{j=1}^{i_{n-1}} g^{-1}(r_j^*)}{i_n-i_{n-1}}),\\
      & \ \varphi^{-1}(\dfrac {\sum_{j=1}^{i_n} \bar E_j+\alpha H_j-\sum_{j=1}^{i_{n-1}} \varphi(r_j^*)}{i_n-i_{n-1}})\},\\
\end{aligned}
\end{equation}
where $i_0=0$, and
\vspace{-0.2cm}
\begin{equation}
\begin{aligned}
      & \  i_n= \underset{i_{n-1}<i\leq N}{\operatorname{argmin}}\{ g(\dfrac {\sum_{j=1}^{i_n}E_j-\sum_{j=1}^{i_{n-1}} g^{-1}(r_j^*)}{i_n-i_{n-1}}),\\
      & \ \varphi^{-1}(\dfrac {\sum_{j=1}^{i_n} \bar E_j+\alpha H_j-\sum_{j=1}^{i_{n-1}} \varphi(r_j^*)}{i_n-i_{n-1}})\}.\\
\end{aligned}
\end{equation}
\end{theorem}

\subsection{Full power transmitter and non-battery receiver}
 In this scenario, there is no limitation on transmitter's power consumption and it may be better for the helper to store its harvested energy and transfer to the receiver later. Intuitively speaking, this problem reduces to a problem with \emph{virtual} receiver. That is, helper's energy is transferred to the receiver instantaneously and we take waterfilling algorithm on total receiver energy taking into account the receiver energy saving inability constraint. If $\{q_i\}_{i=1}^N$ denotes receiver power consumption, the problem is
 \vspace{-0.2cm}
\begin{subequations}
\small{\begin{align}
      & \  \underset{\boldsymbol{q},\boldsymbol{\delta}}{\max} \hspace{0.3cm} \sum_{i=1}^N \varphi^{-1}(q_i) \\
      & \  \text{s.t.}\hspace{0.3cm} q_i\geq \bar E_i,  \hspace{0.5cm} \forall i\\
      & \  \sum_{i=1}^j q_i\leq \sum_{i=1}^j\bar E_ i+ \alpha \delta_i, \hspace{0.5cm} \forall j \\
       & \  \sum_{i=1}^j \delta_i\leq \sum_{i=1}^j H_i,  \hspace{0.5cm} \forall j
\end{align}}
\end{subequations}
where (6b) is the receiver energy saving inability constraint. Moreover, (6c) and (6d) refer to the receiver and helper energy causality constraints, respectively. It is worth mentioning that saving energy occurs at helper only. Using Lemma 1, one can replace (6c) and (6d) by $\sum_{i=1}^j q_i\leq \sum_{i=1}^j\bar E_ i+ \alpha H_i$ for all $j$, which is named (6c$^\prime$). We call (6c$^\prime$) the virtual receiver energy causality constraint which assume that the receiver is virtually initiated with $\bar E_i+\alpha H_i$ amount of energy at $i^{th}$ slot and we use it to find the optimal power transfer policy$ \{\delta_i^*\}_{i=1}^N$. Therefore (6a-d) can be reformulated as (6a,b,c$^\prime$). If $\{q_i^{*}\}_{i=1}^{N}$ denotes the optimal power policy for (6), we characterize the optimal solution (6) in the three following lemmas.

\begin{lemma} If there exist two time slots $m$ and $n$ where $m<n$, such that $q_n^*<q_m^*$, then $q_m^*=\bar E_m$.
\end{lemma}
\begin{proof} Assuming that $q_m^*> \bar E_m$, there is sufficiently small amount of energy, $\epsilon$, which can be saved at $m^{th}$ slot without violating (6b) and to be used at $n^{th}$ time slot. Then, we have
\begin{equation}
\begin{aligned}
      & \  \varphi^{-1}(q_m^*-\epsilon)+\varphi^{-1}(q_n^*+\epsilon)>\  \varphi^{-1}(q_m^*)+\varphi^{-1}(q_n^*). \\
\end{aligned}
\end{equation}
Since $\varphi^{-1}$(.) is a concave function, the left side of (7) leads to larger data transmission which contradicts the problem solution optimality. Thus, $q_m^*=\bar E_m$.
\end{proof}
 We assume that $\{\hat q_i\}_{i=1}^N$ denotes the optimal power policy for problem \{(6a),(6c$^\prime$)\}. $\{\hat q_i\}_{i=1}^N$ can be calculated by forward waterfilling algorithm. Then, we have lemma 5 and 6.
\begin{lemma} If $\hat q_k<\bar E_k$ for some $k$, then $q_k^*=\bar E_k$
\end{lemma}

\begin{proof} Considering the  waterfiling algorithm for \{(6a),(6c$^\prime$)\}, the term $\hat q_k<\bar E_k$ means that to derive the optimum policy for  \{(6a),(6c$^\prime$)\}, the virtual receiver saves $\bar E_k+\alpha H_k-\hat q_k$ amount of energy at $k^{th}$ slot which is larger than $\alpha H_k$. Since the virtual receiver can save at most $\alpha H_k$ amount of energy at $k^{th}$ slot, the amount of energy to be saved would be exactly $\alpha H_k$ because of the problem (6) convexity. Hence, $q_k^*=\bar E_k$.
\end{proof}
Then, it can be shown that ignoring the $k^{th}$ slot at which $\hat p_k<\bar E_k$ from waterfilling will result in non-increasing water level in whole time slots except the $k^{th}$ slot.
\begin{lemma} If $\hat q_k<\bar E_k$ for some $k$ and $\{\tilde q_i\}_{i=1}^N$ denotes optimal policy at (8), then  $\tilde q_i \leq \hat q_i$ for $1\leq i \leq N$ and $i \neq k$.
\begin{subequations}
\small{\begin{align}
      & \  \underset{\boldsymbol{q},\boldsymbol{\delta}}{\max} \hspace{0.3cm} \sum_{i=1,i\neq k}^N \varphi^{-1}(q_i) \\
      & \  \text{s.t.}\sum_{i=1,i\neq k}^j q_i\leq \sum_{i=1,i\neq k}^j\bar E_ i+ \alpha \delta_i, \hspace{0.5cm} \forall j \\
       & \  \sum_{i=1}^j \delta_i\leq \sum_{i=1}^j H_i,  \hspace{0.5cm} \forall j \\
        & \  \delta_k=0.
\end{align}}
\end{subequations}
\end{lemma}

\begin{proof} The optimal solution for \{(6a),(6c$^\prime$)\} is equivalent to problem (9) as

\vspace{-0.3cm}
\begin{subequations}
\small{\begin{align}
      & \  \underset{\boldsymbol{q}}{\max} \hspace{0.3cm} \sum_{i=1,i\neq k}^N \varphi^{-1}(q_i) \\
      & \  \text{s.t.}\sum_{i=1,i\neq k}^j q_i\leq \sum_{i=1,i\neq k}^j\bar E_ i+ \alpha H'_i, \hspace{0.5cm} \forall j
\end{align}}
\end{subequations}
where $\{H'_i=H_i\}_{i=1,i\neq k,k+1}^N$ and $H'_{k+1}=H_{k+1}+$$H_{k}+\dfrac{\bar E_k}{\alpha}-\dfrac{\hat q_k}{\alpha}$. While using lemma 1, (8) is equivalent to \{(9a),(9b$^\prime$)\} where (9b$^\prime$) is $\sum_{i=1,i\neq k}^j q_i\leq \sum_{i=1,i\neq k}^j\bar E_ i+ \alpha H''_i$ for all $j$ in which $\{H''_i=H_i\}_{i=1,i\neq k,k+1}^N$ and $H''_{k+1}=H_{k+1}+H_{k}$.
Since $\hat q_k<\bar E_k$, it can be easily concluded that $\{H'_i=H''_i\}_{i=1,i\neq k,k+1}^N$ and $H'_{k+1}> H''_{k+1}$ meaning that both are the same ordinary waterfilling problem through whole time slots except $k^{th}$ time slot,  but more power available at $k+1^{th}$ slot at (9). Hence, the optimum power for \{(9a),(9b$^\prime$)\} will never be larger than (9) at any remaining slot. Equivalently, the optimal power policy for (8) will never be larger than the optimal power policy for \{(6a),(6c$^\prime$)\} i.e. $\{\tilde q_i \leq \hat q_i\}_{i=1,i\neq k}^N$.
\end{proof}

 Lemma 5 and 6 expresses that if $\hat q_k < \bar E_k$, we can save whole energy at $k^{th}$ slot at helper. This omitting reduces water level and never contradicts our decision about omitting that slot. This reduction in water level may cause some new slot $l$ to become less than $\bar E_l$. Hence, to calculate $\{q_i^*\}_{i=1}^N$, we propose iterative steps as presented in algorithm 1, with generally a few iteration. Benefited from this parameter separation, we have
 \begin{equation}
\begin{aligned}\small
      \  q_i^*=\bar S_i, \text{ and }  \delta_i^*=\dfrac{q_i^*-\bar E_i}{\alpha}, \forall i.
\end{aligned}
\end{equation}

\subsection{Battery transmitter and non-battery EH receiver}
Here, due to incapability of the receiver's energy saving, the transmitter and helper energy saving policy should be performed to maximize the total data transmission rate. In this case also, we can assume that the helper has transferred it's harvested energy instantaneously to the receiver and apply waterfilling solution taking into account the receiver energy saving inability and the transmitter energy management.
\vspace{-0.15cm}
\begin{subequations}
\small{\begin{align}
      & \  \underset{\boldsymbol{r},\boldsymbol{\bar r},\boldsymbol{\delta}}{\max} \hspace{0.3cm} \sum_{i=1}^N r_i \\
      & \  \text{s.t.}\hspace{0.3cm}\sum_{i=1}^j g^{-1}(r_i)\leq \sum_{i=1}^j E_i,  \hspace{0.5cm}\forall j \\
      & \  \sum_{i=1}^j \varphi(\bar r_i)\leq \sum_{i=1}^j\bar E_ i+ \alpha \delta_i, \hspace{0.5cm}\forall j\\
      & \  \sum_{i=1}^j \delta_i\leq \sum_{i=1}^j H_i,  \hspace{0.5cm}\forall j \\
      & \  \varphi(\bar r_j)\geq \bar E_j,  \hspace{0.5cm}\forall j\\
      & \  r_j\leq \bar r_j. \hspace{0.5cm}\forall j
\end{align}}
\end{subequations}

\vspace{-0.5cm}
Any power saving policy at helper imposes a set of maximum constraints on transmitter power management at each time slot. (11e) denotes energy saving inability at receiver. From lemma 1, one can substitute (11c) and (11d) with $\sum_{i=1}^j \varphi(\bar r_i)\leq \sum_{i=1}^j\bar E_ i+ \alpha H_i, \forall j$ which is denoted by (11c$^\prime$). As before, we assumed that the receiver energy is initiated with $\bar E_ i+ \alpha H_i$ at $i^{th}$ slot which help us calculate optimal energy transfer policy $\{\delta_i^*\}_{i=1}^N$ using $\{\bar r_i^*,r_i^*\}_{i=1}^N$. To solve (11), we decompose the problem into outer and inner problem. We show the inner problem (11a,b,f) by the function $\{r_i\}_{i=1}^N=\textbf{R}_g(\{E_i\}_{i=1}^N,\{\bar r_i\}_{i=1}^N)$ which maximize $\sum_{i=1}^N r_i$. For fixed $\{\bar r_i\}_{i=1}^N$, the function $\textbf{R}_g$ is the waterfilling algorithm on  the first argument $\{E_i\}_{i=1}^N$, taking into account the maximum constrains $\{g^{-1}(\bar r_i)\}_{i=1}^N$ at each slot, where the optimal solution is considered in \cite{44} and we avoid representing the optimal algorithm. Moreover, we specify the outer problem on (11) as $\{\bar r_i\}_{i=1}^N=\textbf{F}_\varphi(\{ \bar E_i\}_{i=1}^N, \{H_i\}_{i=1}^N)$, where $\textbf{F}_\varphi$ is any feasible $\{\bar r_i\}_{i=1}^N$ under two constraints (11c$^\prime$,e). So, the optimal rates at (11) is given by

\begin{equation}
\begin{aligned}\small
      & \   \textbf{R}_g(\{E_i\}_{i=1}^N,\textbf{F}_\varphi(\{ \bar E_i\}_{i=1}^N, \{H_i\}_{i=1}^N)) \\
\end{aligned}
\end{equation}
 We denote the result for (12), by $\{\hat r_i\}_{i=1}^N$ when $\textbf{F}_\varphi$ is fixed as $\textbf{F}_\varphi^*(\{ \bar E_i\}_{i=1}^N, \{H_i\}_{i=1}^N)=\{\varphi^{-1}(\bar S_i)\}_{i=1}^N$, where $\{\bar S_i\}_{i=1}^N$ is derived from algorithm 1. Besides, $\{r_i^*\}_{i=1}^N$ denote the optimal policy at (11). Then, we have the following lemma.
\begin{lemma} $\hat r_1=r_1^*$.
\end{lemma}
\begin{proof} Algorithm 1 is the waterfilling algorithm to calculate $\{\bar S_i\}_{i=1}^N$, in terms of $\{\bar E_i\}_{i=1}^N$ and $\{H_i\}_{i=1}^N$ which waterfill helper energy in the way that non-battery receiver is enabled to decode maximum data with the assumption that transmitter is full power. We assume that $\textbf{F}_\varphi^*$ does not follow algorithm 1 for power management i.e. $\{\bar r_i^*\}_{i=1}^{N}=\textbf{F}_\varphi^*(\{ \bar E_i\}_{i=1}^N, \{H_i\}_{i=1}^N)\neq \{\varphi^{-1}(\bar S_i)\}_{i=1}^N$ and specifically $\bar r_1^*\neq \varphi^{-1}(\bar S_1)$. Then, one of two following assumption must happen. First, in contradiction to lemma 4, there must be some time slot $n>1$ such that $\bar r_1^*>\bar r_n^*$ and $\varphi(\bar r_1^*)>\bar E_1$. So, noting that $\{r_i^*\}_{i=1}^N=\textbf{R}_g(\{E_i\}_{i=1}^N, \{\bar r_i^*\}_{i=1}^N)$, we would have two cases.

 The first case of the first assumption is $r_1^*<\bar r_1^*$. So, there would be sufficiently small amount of energy $\epsilon$, which can be saved at first slot for $n^{th}$ slot at helper in the way that we could have new maximum limitation set i.e $\{\bar r'_i\}_{i=1}^N=\{\varphi^{-1}(\varphi(\bar r_1^*)-\epsilon),\varphi^{-1}(\varphi(\bar r_n^*)+\epsilon)\}\cup \{\bar r_i\}_{i=2,i\neq n}^N$ as the output for $\textbf{F}_\varphi^*$. Noting that $\varphi$(.) is an increasing function, the new set $\{\bar r'_i\}_{i=1}^N$ imposes larger maximum constraint at time slot $n$ at transmitter, without constraining transmitter rate at time slot $1$, which contradicts $\textbf{F}_\varphi^*$ optimality.

  Now, we consider the second case of the first assumption, $r_1^*=\bar r_1^*$. Then, it can be shown that $r_n^*=\bar r_n^*$ because of the function $\textbf{R}_g$ features in having two time slot maximum rate constraint $\bar r_1^*$ and $\bar r_n^*$ where $\bar r_1^*>\bar r_n^*$ because of $g$(.) concavity and we avoid presenting the complete proof. In this case, considering the new set $\{\bar r'_i\}_{i=1}^N$ as defined above again, let the transmitter  save $\epsilon'=g^{-1}(r_1^*)-g^{-1}(\varphi^{-1}(\varphi(r_1^*)-\epsilon))$ amount of energy at first slot for the $n^{th}$ slot. Defining $\{ r'_i\}_{i=1}^N=\textbf{R}_g(\{E_i\}_{i=1}^N,\{\bar r'_i\}_{i=1}^N)$, it can proven that $r'_1=\bar r'_1$ and for the $n^{th}$ slot, we have
 \begin{equation}
\begin{aligned}\small
      & \   r'_n=\text{min}\{g(g^{-1}(r_n^*)+\epsilon'),\varphi^{-1}(\varphi(\bar r_n^*)+\epsilon)\} \\
\end{aligned}
\end{equation}
 If $ r'_n$ is equal to the first argument of (13), it can be shown that $r'_1+ r'_n=g(g^{-1}(r_1^*)-\epsilon')+g(g^{-1}(r_n^*)+\epsilon')>r_1^*+\bar r_n^*$, because $g(.)$ is a concave function. Similarly, if $r'_n$ is equal to the second argument, noting that $r'_1=\bar r'_1$, it can be shown that $r'_1+ r'_n=\varphi^{-1}(\varphi(r_1^*)-\epsilon)+\varphi^{-1}(\varphi(r_n^*)+\epsilon)>r_1^*+\bar r_n^*$. Thus the new policy for $\textbf{F}_\varphi^*$ i.e. $\{\bar r'_i\}_{i=1}^N=\{\varphi^{-1}(\varphi(\bar r_1^*)-\epsilon),\varphi^{-1}(\varphi(\bar r_n^*)+\epsilon)\}\cup \{\bar r_i\}_{i=2,i\neq n}^N$ let the function $\textbf{R}_g$ obtain $\{r'_i\}_{i=1}^N$ where $\sum_{i=1}^N r'_i>\sum _{i=1}^N r_i^*$ which contradicts initial assumption for the optimality of $\textbf{F}_\varphi^*$.

 The second assumption is that for some slot $n>1$, $\bar r_1^*<\bar r_n^*$ and it is feasible to turn  some energy from $n^{th}$ time slot back to the first time slot at helper. Non-optimality of this case can be concluded in a similar way as the first assumption and we discard the proof to avoid repetition. Hence, in case the function $\textbf{F}_\varphi^*(\{ \bar E_i\}_{i=1}^N ,\{H_i\}_{i=1}^N)$ follows the waterfilling algorithm 1 to provide the optimum maximum energy constraint for $\textbf{R}_g$, the function $\textbf{R}_g(\{E_i\}_{i=1}^N,\textbf{F}_\varphi^*(\{ \bar E_i\}_{i=1}^N, \{H_i\}_{i=1}^N))$ provides the optimum rate at first time slot at (11) i.e $r_1^*$.
 \end{proof}
 Lemma 7 is the main idea of  algorithm 2 to calculate (11).
\subsection{EH transmitter and receiver with no batteries}
In this scenario, only helper can save energy for upcoming slots. If the constraint, $r_j\leq  g(E_j),\forall j$ denotes the transmitter energy saving inability and is shown by (11b$^\prime$), the problem formulation for this scenario becomes (11a,b$^\prime$-f). 

Intuitively speaking, we assume that helper transfers the harvested energy instantaneously to the receiver which is called a virtual receiver. Note that the receiver energy saving inability imposes minimum energy constraint at the virtual receiver at any time slot. Besides, harvested energy at non-battery transmitter imposes maximum energy constraint on the virtual receiver at all slots. Using lemma 1 to 6, we present algorithm 3 to find the optimal policy for (11a,b$^\prime$-f) as
\begin{equation}
\begin{aligned}
      \small { \  r_i^*=\text{min}(g(E_i),\varphi^{-1}(\bar S_i)), \forall i. }
\end{aligned}
\end{equation}

\section{Numerical Results and Conclusion}
We present numerical examples for the second and third scenarios. Suppose that $\alpha$=0.7 and $\varphi^{-1}(.)=g(.)$. We initialize $\bf{\bar E}=$[5,8,3] and $\bf{H}=$[7,1,2] in three time slots, i.e., $N=3$. Using algorithm 1 for the second scenario (\emph{full power transmitter}), we derive $\bf{q}^*$=[7.5,8,7.5] as the optimal power consumption at the receiver. It can be seen that the helper do not transfer energy at the second time slot to the receiver. For the third scenario, we initialize $\bf{\bar E}$ and $\bf{H}$ as before and set $\bf{E}$=[6.5,13.5,9]. By use of algorithm 2, we calculate the optimal policy as $\bf{r}^*$=$g$([6.5,8.25,8.25]). In consistent with lemma 7, we considered $\bf{\bar S}=\bf{q}^*$ as the maximum constraint on the transmitter for all the time slots to determine the first time slot. Then, according to algorithm 2, one unit of energy is saved at $\bf{\bar S}^*$ at the first time slot for the two remaining ones.

In summary, we studied maximizing data rate transmission over a point-to-point AWGN channel with EH transmitter, receiver and energy cooperating helper. The helper efficiently managed and transferd its harvested energy to the receiver, since the receiver needs sufficient power to decode the message. We analyzed four scenarios; When both EH transmitter and receiver have batteries, we presented a closed form optimal solution. For other scenarios (full power transmitter and non-battery receiver, battery transmitter and non-battery receiver, and EH transmitter and receiver with no batteries), we derived iterative algorithms to achieve the optimal power policies.

%------------------------------------------------------------------------------------------------------------
%-----------------------------------------------------------------

% -------------------------------------------------------------------------------------

%\newpage
 \noindent\makebox[\linewidth]{\rule{9cm}{0.9pt}}
 \small { \textbf{Iterative Algorithms}

 \vspace{-0.2cm}
 \noindent\makebox[\linewidth]{\rule{9cm}{0.3pt}}
 \textbf{Initialization}\\
1: Initialize the transmitter (algorithm 2,3), receiver and helper

 energies; $\{E_i,\bar E_i,H_i\}_{i=1}^N$

\hspace{-0.5cm} \textbf{Define functions}\\
2: $g(.)$ and $\varphi(.)$;   transmission rate and decoding cost function \\

\vspace{-0.3cm}
 \hspace{-0.5cm} \textbf{Define procedure}\\
3:  $A(\{L'_i\}_{i\in \mathcal{X}})=\textbf{R}_{A}$($\{L_i\}_{i\in \mathcal{X}},\{r_i\}_{i\in \mathcal{X}}$); calculate $\{L'_i\}_{i\in \mathcal{X}}$ by

waterfilling the energies $\{L_i\}_{i\in \mathcal{X}}$,  taking into account the

maximum energy constraints $\{A^{-1}(r_i)\}_{i\in \mathcal{X}}$ at all slots.\\

\vspace{-0.3cm}
\hspace{-0.3cm}4: Choose algorithm 1,2, or 3.

\vspace{-0.2cm}
\noindent\makebox[\linewidth]{\rule{9cm}{0.3pt}}
\textbf{Algorithm 1:} (\emph{Full power transmitter and non-battery receiver})

\vspace{-0.2cm}
\noindent\makebox[\linewidth]{\rule{9cm}{0.3pt}}
5: Initialize the arrays; $\{S_i,\bar S_i=0\}_{i=1}^N$, the sets; $\nabla= \phi$, $ \bar \nabla=\{0\}$\\
\hspace{-0.5cm} 6:  \textbf{while} $\bar \nabla \neq \phi$ \textbf{do}\\
\hspace{-0.3cm} 7:  \hspace{0.1cm} Update $H_i$ by sending the whole water in $\nabla$ bins at helper

   \hspace{0.15cm} forward to the first bin which is not $\nabla$ member\\
8: \hspace{0.25cm}$S_i \leftarrow \bar E_i+\alpha H_i$ for all $i$\\
9: \hspace{0.15cm} Waterfill $\{S_i\}_{i=1,i\notin \nabla}^N$ and let the result be $\{\bar S_i\}_{i=1,i\notin \nabla}^N$\\
10:\hspace{0.25cm}Regenerate $\bar \nabla$ by finding the bins at which $\bar S_i<\bar E_i$  and put

\hspace{0.25cm} those slot numbers to the set $\bar \nabla$ $\hspace{0.2cm}$\\
11:\hspace{0.12cm} $\nabla\leftarrow \nabla \cup \bar \nabla$\\
12: \textbf{end while}\\
13:$\{\bar S_i\}_{i\in \nabla}\leftarrow \{\bar E_i\}_{i\in \nabla}$\\
14: \textbf{Return}

\vspace{-0.2cm}
\noindent\makebox[\linewidth]{\rule{9cm}{0.9pt}}
\textbf{Algorithm 2:} (\emph{Battery transmitter and non-battery receiver})

\vspace{-0.2cm}
\noindent\makebox[\linewidth]{\rule{9cm}{0.3pt}}
5: Set $k=1$\\
6: \textbf{while} $k\leq N$ \textbf{do}\\
7: $\hspace{0.1cm}$ $\{\bar r_i\}_{i=k}^N \leftarrow \{\varphi^{-1}(\bar  S_i)\}_{i=k}^N$ using  algorithm 1 for bins $k$ to $N$\\
8: \hspace{0.1cm} $\{r_i\}_{i=k}^N \leftarrow \textbf{R}_{g}(\{E_i\}_{i=k}^N , \{\bar r_i\}_{i=k}^N)$\\
9: $\hspace{0.10cm}$ $\textbf{if}$ $\bar E_k<\varphi(r_k)$\\
10: $\hspace{0.15cm}$ Update $\{H_i\}_{i=1}^N$ by saving $\dfrac{[\varphi(\bar r_k)-\varphi(r_k)]^+}{\alpha}$ amount of

$\hspace{0.4cm}$ energy at helper at $k^{th}$  slot to $k+1^{th}$ slot.\\
11: $\hspace{0.01cm}$ $\textbf{else}$\\
12: $\hspace{0.2cm}$ Update $\{H_i\}_{i=1}^N$ by saving $\dfrac{[\varphi(\bar r_k)-\bar E_k]^+}{\alpha}$ amount of

 $\hspace{0.33cm}$ energy at helper at $k^{th}$ slot to $k+1^{th}$ slot.\\
13:$\hspace{0.1cm}$ $\textbf{end if}$\\
14: $k \leftarrow k+1$\\
15: \textbf{end while}\\
16: $\{r_i^*\}_{i=1}^N \leftarrow \textbf{R}( \{E_i\}_{i=1}^N, \{\varphi^{-1}(\bar E_i+\alpha H_i)\}_{i=1}^N)$\\
17: \textbf{Return}

\vspace{-0.2cm}
\noindent\makebox[\linewidth]{\rule{9cm}{0.9pt}}
\textbf{Algorithm 3:} (\emph{EH transmitter and receiver with no batteries})

 \vspace{-0.2cm}
 \noindent\makebox[\linewidth]{\rule{9cm}{0.3pt}}
5: Initialize the arrays; $\{S_i,\bar S_i=0\}_{i=1}^N$, the sets; $\nabla= \phi$, $ \bar \nabla=\{0\}$\\
\hspace{-0.5cm} 4: \textbf{while} $\bar \nabla \neq \phi$ \textbf{do}\\
6:  \hspace{0.1cm} Update $H_i$ by sending the whole water in $\nabla$ bins at helper

   \hspace{0.2cm}forward to the first bin which is not $\nabla$ member\\
7: \hspace{0.2cm}$S_i \leftarrow \bar E_i+\alpha H_i$ for all $i$\\
8: \hspace{0.1cm} $\{\bar S_i\}_{i=1, i\notin \nabla}^N  \leftarrow \varphi(\textbf{R}_{\varphi^{-1}}(\{S_i\}_{i=1, i\notin \nabla}^N,\{g(E_i)\}_{i=1, i\notin \nabla}^N))$\\
9:\hspace{0.32cm}Regenerate $\bar \nabla$ by finding the bins at which $\bar S_i<\bar E_i$  and put

\hspace{0.25cm}those slot numbers to the set $\bar \nabla$ $\hspace{0.2cm}$\\
10\hspace{0.17cm} $\nabla\leftarrow \nabla \cup \bar \nabla$\\
11: \textbf{end while}\\
12:$\{\bar S_i\}_{i\in \nabla}\leftarrow \{\bar E_i\}_{i\in \nabla}$\\
13: \textbf{Return} }

\vspace{-0.24cm}
\noindent\makebox[\linewidth]{\rule{9cm}{0.9pt}}

\end{document}